\documentclass[envcountsame, oribibl]{llncs} %,envcountsect

\usepackage{times,amsmath,bbold, pgf, tikz, verbatim, stmaryrd}
\usetikzlibrary{arrows,automata}

\DeclareMathOperator{\Acc}{Acc}
\DeclareMathOperator{\Val}{val}
\DeclareMathOperator{\Weight}{Weight}
\DeclareMathOperator{\Mon}{Mon}
\DeclareMathOperator{\Free}{Free}
\DeclareMathOperator{\dom}{dom}

\DeclareMathOperator{\supp}{supp}

\newcommand{\zero}{\mathbb 0}
\newcommand{\one}{\mathbb 1}
\newcommand{\mV}{\mathcal V}

\newcommand{\MSO}{{\bf MSO}}
\newcommand{\li}{\langle \! \langle}
\newcommand{\re}{\rangle \! \rangle}
\newcommand{\Mult}{\mathbb N \langle M \rangle}
\newcommand{\valo}{\Val^{\omega}}
\newcommand{\sov}{\Sigma^{\omega}_{\mV}}
\newcommand{\so}{\Sigma^{\omega}}

\newcommand{\mwMA}{{\bf mwMA}}
\newcommand{\mwA}{{\bf mwA}}
\newcommand{\mwMSO}{{\bf mwMSO}}
\newcommand{\wMA}{{\bf wMA}}
\newcommand{\wA}{{\bf wA}}
\newcommand{\wMSO}{{\bf wMSO}}

\newcommand{\eN}{\overline{\mathbb N}}
\newcommand{\bli}{/\! \!/}
\newcommand{\bre}{/ \! \! /}
\newcommand{\cre}{ / \! \! /}

\begin{document}

\title{Multi-weighted Automata and MSO Logic\thanks{The final publication is available at link.springer.com; DOI: 10.1007/978-3-642-38536-0\_36}}
\author{Manfred Droste and Vitaly Perevoshchikov\thanks{Partially supported by  DFG Graduiertenkolleg 1763 (QuantLA)}}
\institute{Universit\"at Leipzig,  Institut f\"ur Informatik, \\ 
04109 Leipzig, Germany\\
\email{\{droste,perev\}@informatik.uni-leipzig.de}
}

\maketitle

\begin{abstract}
Weighted automata are non-deterministic automata where the transitions are equipped with weights. They can model quantitative aspects of systems like costs or energy consumption. The value of a run can be computed, for example, as the maximum, limit average, or discounted sum of transition weights. In multi-weighted automata, transitions carry several weights and can model, for example, the ratio between rewards and costs, or the efficiency of use of a primary resource under some upper bound constraint on a secondary resource. Here, we introduce a general model for multi-weighted automata as well as a multi-weighted MSO logic.
In our main results, we show that this multi-weighted MSO logic and multi-weighted auto\-mata are expressively equivalent both for finite and infinite words. The translation process is effective, leading to decidability results for our multi-weighted MSO logic.
\end{abstract}

\section{Introduction}

Recently, multi-priced timed automata \cite{BBL04, BBL08,  FLT11, LR05} have received much attention for real-time systems. These automata extend priced timed automata by featuring several price parameters. This permits to compute objectives like the optimal ratio between rewards and costs \cite{BBL04, BBL08}, or the optimal consumption of several resources where more than one resource must be restricted \cite{LR05}. Arising from the model of timed automata, the multi-weighted setting  has also attracted much notice for classical non-deterministic automata \cite{BJLLS12, BGH09,  FJLS11, FGR11}. 

The goal of the present paper is to develop a multi-weighted monadic second order (MSO) logic and to show that it  is expressively equivalent to multi-weighted automata. 

B\"uchi's and Elgot's fundamental theorems \cite{Buc60, Elg61} established the expressive equivalence of finite automata and MSO logic. Weighted MSO logic with weights taken from an arbitrary semiring was introduced in \cite{DG07, DG09} and it was shown that a fragment of this weighted logic and semiring-weighted automata on finite and infinite words have the same expressive power \cite{DG09}. Chatterjee, Doyen, and Henzinger \cite{CDH08, CDH10} investigated weighted automata modeling the average and long-time behavior of systems. The behavior of such automata cannot be described by semiring-weighted automata. In \cite{DM11, DM12}, valuation monoids were presented to model the quantitative behaviors of these automata. Their  logical characterization was given in \cite{DM12}. 
In this paper, we establish, both for finite and infinite words, the B\"uchi-type result for multi-weighted automata; these do not fit into the framework of other weighted automata like semiring automata \cite{BR88, DKV09, KS86, SS78}, or even valuation monoid automata \cite{DM11, DM12}.

First, we develop a general model for multi-weighted automata which incorporates several multi-weighted settings from the literature. Next, we define a multi-weighted MSO logic by extending the classical MSO logic with constants which could be tuples of weights. The semantics of formulas should be single weights (not tuples of weights). Different from weighted MSO logics over semirings or valuation monoids, this makes it impossible to define the semantics inductively on the structure of an MSO formula. Instead, for finite words, we introduce an intermediate semantics which maps each word to a finite multiset containing tuples of weights. The semantics of a formula is then defined by applying to the multiset semantics an operator which evaluates a multiset to a single value. Our B\"uchi-type result for multi-weighted automata on finite words is established by reducing it to the corresponding result of \cite{DM12} for the product valuation monoid of finite multisets.

In the case of infinite words, it is usually not possible to collect all the information about weights of paths in finite multisets. Therefore, we cannot directly reduce the desired result to the proof given in \cite{DM12} for infinite words. But we can use the result of \cite{DM12} to translate each multi-weighted formula of our logic into an automaton over the product $\omega$-valuation monoid of multisets, and we show that the weights of transitions in this automaton satisfy certain properties which allow us to translate it into a multi-weighted automaton.

All our automata constructions are effective. Thus, decision problems for multi-weighted logic can be reduced to decision problems of multi-weighted automata. Some of these problems for automata can be solved whereas for others the details still have to be explored.

\section{Multi-weighted Automata on Finite Words}

The model of {\it multi-weighted} (or {\it multi-priced}) automata is an extension of the model of weighted automata over semirings \cite{BR88, DKV09,KS86,SS78} and valuation monoids \cite{DM11, DM12} by featuring several price parameters. In the literature, different situations of the behaviors of multi-weighted automata  were considered (cf. \cite{BJLLS12, BGH09, BBL04, BBL08,  FJLS11, FLT11, FGR11, LR05}) to model the consumption of several resources. For instance, the model of multi-priced timed automata introduced in \cite{BBL04} permits to describe the optimal ratio between accumulated rewards and accumulated costs of transitions. In this section, we introduce a general model to describe the behaviors of multi-weighted automata on finite words. 

Consider an automaton in which every transition carries a reward and a cost. For paths of transitions, we are interested in the ratio between accumulated rewards and accumulated costs. The automaton should assign to each word the maximal reward-cost ratio of accepting paths on $w$. The idea is to model the weights by elements of the set ${M = \mathbb R \times \mathbb R_{\ge 0}}$. We use a valuation function $\Val: M^+ \to M$ to associate to each sequence of such weights a single weight in $M$. Since our automata are nondeterministic and a word may have several accepting paths, we obtain a multiset of weights of these paths, hence a multiset of elements from $M$. We use an evaluator function $\Phi$ which associates to each multiset of $M$ a single value. The mapping $\Phi$ can be considered as a general summation operator. Now we turn to formal definitions.

To cover also the later case of infinite words, we let $\overline{\mathbb N} = \mathbb N \cup \{\infty\}$. Let $M$ be a set. A {\it multiset} over $M$ is a mapping $r: M \to \overline{\mathbb N}$. For each $m \in M$, $r(m)$ is the number of copies of $m$ in $r$. We let $\supp(r) = \{m \in M \; | \; r(m) \neq 0\}$, the {\it support} of $r$. We say that a multiset $r$ is {\it finite} if $\supp(r)$ is finite and $\infty \notin r(M)$. We denote the collection of all multisets by $\eN \li M \re$ and the collection of all finite multisets by $\Mult$.
\begin{definition}
\label{Def:VS}
Let $K$ be a set. A {\rm $K$-valuation structure} $(M, \Val, \Phi)$ consists of a set $M$, a {\rm valuation function} $\Val: M^+ \to M$ with $\Val(m) = m$ for all $m \in M$, and an {\rm evaluator function} ${\Phi: \Mult \to K}$.
\end{definition}
A {\em nondeterministic automaton} over an alphabet $\Sigma$ is a tuple ${\mathcal A = (Q, I, T, F)}$ where $Q$ is a set of {\it states}, $I, F \subseteq Q$  are sets of {\it initial} resp. {\it final states} and ${T \subseteq Q \times \Sigma \times Q}$ is a {\it transition relation}. {\it Finite paths} $\pi = (t_i)_{0 \le i \le n}$ of $\mathcal A$ are defined as usual as finite sequences of matching transitions, say $t_i = (q_{i}, a_{i}, q_{i+1})$. Then we call the word ${w =  a_0 a_1 ... a_n \in \Sigma^+}$ the {\it label} of the path $\pi$ and $\pi$ a path on $w$. A path is {\it accepting} if it starts in $I$ and ends in $F$. We denote the set of all accepting paths of $\mathcal A$ on $w \in \Sigma^+$ by $\Acc_{\mathcal A}(w)$. 

\begin{definition}
Let $\Sigma$ be an alphabet, $K$ a set and $\mathcal M = (M, \Val, \Phi)$ a $K$-valuation structure. A {\rm multi-weighted automaton} over $\Sigma$ and $\mathcal M$ is a tuple $(Q, I, T, F, \gamma)$ where $(Q, I, T, F)$ is a nondeterministic automaton and $\gamma: T \to M$.
\end{definition}
Let $\mathcal A$ be a multi-weighted automaton over $\Sigma$ and $\mathcal M$, $w \in \Sigma^+$ and $\pi = t_0 ... t_n$ a path on $w$. The {\it weight} of $\pi$ is defined by 
$\Weight_{\mathcal A}(w) = \Val(\gamma(t_i))_{0 \le i \le n}$. Let ${| \mathcal A | (w) \in \Mult}$ be the finite multiset containing the weights of all accepting paths in $\Acc_{\mathcal A}(w)$. Formally, $| \mathcal A | (w)(m) = |\{\pi \in \Acc_{\mathcal A}(w) \; | \; \Weight_{\mathcal A}(\pi) = m\}|$ for all $m \in M$.
The {\it behavior} $||\mathcal A||: \Sigma^+ \to K$ of $\mathcal A$ is defined for all $w \in \Sigma^+$ by $||\mathcal A||(w) = \Phi(|\mathcal A|(w))$.

Note that every weighted automaton over a valuation monoid $(M, +, \Val, \mathbb 0)$ (cf. \cite{DM11, DM12}) can be considered as a multi-weighted automaton over the $K$-valuation structure $(M, \Val, \Phi)$ with $K = M$ and $\Phi: \Mult \to M$ defined by ${\Phi(r) = \sum (m \; | \; m \in \supp(r) \text{ and } 1 \le i \le r(m) )}$ (as usual, $\sum \emptyset = \mathbb 0$). Moreover,  multi-weighted automata extend the model of weighted automata over valuation monoids  in two directions. First, whereas the weights of transitions in multi-weighted automata are taken from $M$, the behavior is a mapping with the codomain $K$ where $K$ and $M$ do not necessarily coincide. Second, we resolve the nondeterminism in multi-weighted automata using an evaluator function $\Phi$ defined on finite multisets. 

Next, we consider several examples how to describe the behavior of multi-weighted automata known from the literature using valuation structures. In each of the three examples below, let $\Sigma$ be an alphabet,  $\mathcal M = (M, \Val, \Phi)$ a $K$-valuation structure, and $\mathcal A$ a multi-weighted automaton over $\Sigma$ and $\mathcal M$.

\begin{example}
\label{Example:F1}
Let $\overline{\mathbb R} = \mathbb R \cup \{-\infty, \infty\}$. 
Let $M = \mathbb R \times \mathbb R_{\ge 0}$, $K = \overline{\mathbb R}$, $\Val((x_1, y_1), ..., (x_k, y_k)) = \left(\sum_{i = 1}^k x_i, \sum_{i = 1}^k y_i \right)$ be the componentwise sum, and $\Phi$ defined by ${\Phi(r) = \max\limits_{(x, y) \in \supp(r)} \frac{x}{y}}$ \,where we put $\frac{x}{0} = \infty$ and $\max (\emptyset) = -\infty$. For instance, for every transition weight $(x, y) \in M$, $x$ might mean the reward and $y$ the cost of the transition. Then $||\mathcal A||(w)$ is the maximal ratio between accumulated rewards and costs of accepting paths on $w$. The ratio setting was considered first for multi-priced timed automata \cite{BBL04, BBL08} and also for nondeterministic automata \cite{BGH09,FGR11}.
\end{example}

\begin{example}
\label{Example:F2}
Let $M = \mathbb R \times \mathbb R$, $K = \mathbb R \cup \{\infty\}$ and $p \in \mathbb R$. Let $\Val$ be as in the previous example and
$\Phi(r) = \min \{x \; | \; (x, y) \in \supp(r) \text{ and } y \le p \}$, for $r \in \Mult$, with $\min(\emptyset) = \infty$. Let $t$ be a transition and $\gamma(t) = (x, y)$. We call $x$ the primary and $y$ the secondary cost. Then $||\mathcal A||(w)$ is the cheapest primary cost of reaching with $w$ some final state  under the given upper bound constraint $p \in \mathbb R$ on the secondary cost. The optimal conditional reachability problem for multi-priced timed automata was studied in \cite{LR05}.
\end{example}

\begin{comment}
\begin{example}
\label{Example:F3}
Let $M = \mathbb R \times \mathbb R_{\ge 0}$, $K = \mathbb R \cup \{\infty\}$ where, for every $(x, t) \in M$, we regard $x$ as a cost and $t$ as a duration. Let $0 < \lambda < 1$ be a discounting factor. We define $\Val$ by 
$\Val((x_1, t_1), ..., (x_k, t_k)) = \left( \sum_{i = 1}^k \lambda^{\tau_{i-1}} \cdot x_i, \tau_k \right)$ where $\tau_0 = 0$, $\tau_{i + 1} = \tau_i + t_{i+1}$ for $i \ge 0$. Let also $\Phi$ be defined by $\Phi(r) = \min\limits_{(x, t) \in \supp(r)} x$. Then $||\mathcal A||(w)$ is the time minimal discounted cost of accepting paths on $w$. Discounting in time was considered in \cite{FL09}.
\end{example}
\end{comment}

\begin{example}
\label{Example:F4}
Let $M = \mathbb R^n$ for some $n \ge 1$, $K = \mathbb R$, and $\Val$ be the component\-wise sum of vectors. We define $\Phi: \Mult \to \mathbb R$ as follows. Let $r \in \Mult$ and $S = \supp(r)$. Then $\Phi(r) = 0$ if $S = \emptyset$ and $\Phi(r) = \frac{\sum_{v \in S}  r(v) \cdot ||v|| }{\sum_{v \in S} r(v)}$ otherwise. Here, for $v = (v_1, ..., v_n)$, ${||v|| = \sqrt{v_1^2 + ... + v_n^2}}$ is the length of $v$. Suppose that $\mathcal A$ controls the movement of some object in $\mathbb R^n$ and each transition $t$ carries the coordinates of the displacement vector of this object. Then, $||\mathcal A||(w)$ is the value of the average displacement of the object after executing $w$.
\end{example}

\section{Multi-weighted MSO Logic on Finite Words}

In this section, we wish to develop a multi-weighted MSO logic where the weight constants are elements of a set $M$. Again, if weight constants are {\it pairs} of a reward and a cost, the semantics of  formulas must reflect the maximal reward-cost ratio setting, so the weights of formulas should be {\it single weights}. Then, there arises a problem to define the semantics function inductively on the structure of a formula as in \cite{DG09, DM12}. We solve this problem in the following way. We associate to each word a multiset of elements of $M$. Here, for disjunction and existential quantification, we use the multiset union. For conjunction, we extend a product operation given on the set $M$ to the Cauchy product of multisets. Similarly, for universal quantification, we extend the valuation function on $M^+$ to $\Mult^+$. 
Then, we use an evaluator function $\Phi$ which associates to each multiset of elements a single value (e.g. the maximal reward-cost ratio of pairs contained in a multiset).

As in the case of weighted MSO logics over product valuation monoids \cite{DM12}, we extend a valuation structure (cf. Definition \ref{Def:VS}) with a unit element and a binary operation in order to define the semantics of atomic formulas and of the conjunction.

\begin{definition}
Let $K$ be a set. A {\rm product $K$-valuation structure ($K$-pv-structure)} $(M, \Val, \diamond, \one, \Phi)$ consists of a $K$-valuation structure $(M, \Val, \Phi)$, a constant $\one \in M$ with $\Val(m \one ... \one) = m$ for $m \in M$, and a multiplication $\diamond: M \times M \to M$ such that $m \diamond \one = \one \diamond m = m$ for all $m \in M$. 
\end{definition}
For the rest of this section, we fix an alphabet $\Sigma$ and a $K$-pv-structure ${\mathcal M = (M, \Val, \diamond, \one, \Phi)}$. Let $V$ be a countable set of first and second order variables. Lower-case letters like $x, y$ denote first order variables whereas capital letters like $X, Y$ etc. denote second order variables. 
The syntax of {\it multi-weighted MSO logic} over $\Sigma$ and $\mathcal M$ is defined as in \cite{BG09} by the grammar:
\begin{align*}
\beta &::= P_a(x) \; | \; x \le y  \; | \; x \in X \; | \; \lnot \beta \; | \; \beta \wedge \beta  \; | \; \forall x \beta \; | \; \forall X \beta \\
\varphi  &::= m \; | \; \beta \; | \; \varphi \vee \varphi \; | \; \varphi \wedge \varphi \; | \; \exists x \varphi \; | \; \forall x \varphi \; | \; \exists X \varphi
\end{align*}
where $a \in \Sigma$, $m \in M$, $x, y, X \in V$. The formulas $\beta$ are called {\it boolean} formulas and the formulas $\varphi$ {\it multi-weighted MSO}-formulas. Note that negation and universal second order quantification are allowed in boolean formulas only. Note also that the boolean formulas have the same expressive power as (unweighted) MSO logic.

The class of {\it almost boolean} formulas over $\Sigma$ and $M$ is the smallest class containing all constants $m \in M$ and all boolean formulas and which is closed under $\wedge$ and $\vee$. A multi-weighted MSO formula $\varphi$ is {\it syntactically restricted} if
whenever it contains a sub-formula $\forall x \psi$, then $\psi$ is almost boolean, and if for every subformula $\varphi_1 \wedge \varphi_2$ of $\varphi$ either both $\varphi_1$ and $\varphi_2$ are almost boolean, or $\varphi_1$ or $\varphi_2$ is boolean.

The set $\Free(\varphi)$ of free variables in $\varphi$ is defined as usual. For $w \in \Sigma^+$,  let ${\dom(w) = \{0, ..., |w|\text{-}1\}}$. Let $\mV$ be a finite set of variables with $\Free(\varphi) \subseteq \mV$. A $(\mathcal V, w)$-{\it assignment} is a mapping $\sigma: \mathcal V \to \dom(w) \cup 2^{\dom(w)}$ where every first order variable is mapped to an element of $\dom(w)$ and every second order variable to a subset of $\dom(w)$. The update $\sigma[x / i]$ for $i \in \dom(w)$ is defined as:
$\sigma[x / i](x) = i$ and $\sigma[x / i] |_{\mathcal V \setminus \{x\}} = \sigma |_{\mathcal V \setminus \{x\}}$. The update for second order variables can be defined similarly. Each pair $(w, \sigma)$ of a word and $(\mathcal V, w)$-assignment can be encoded as a word over the extended alphabet $\Sigma_\mathcal V = \Sigma \times \{0,1\}^\mathcal V$. Note that a word $(w, \sigma) \in \Sigma_\mathcal V^+$ represents an assignment if and only if, for every first order variable in $\mathcal V$, the corresponding row in the extended word contains exactly one 1; then $(w, \sigma)$ is called {\it valid}. The set of all valid words in $\Sigma_\mathcal V ^+$ is denoted by $\mathcal N_\mathcal V$. We also denote by $\Sigma_{\varphi}$ the alphabet $\Sigma_{\Free(\varphi)}$.

Consider again the collection $\mathbb N \langle M \rangle$ of all finite multisets over $M$. 
Here, we consider the set of natural numbers as the semiring $(\mathbb N, +, \cdot, 0, 1)$ where $+$ and $\cdot$ are usual addition and multiplication.
The {\it union} $(r_1 \oplus r_2) \in \mathbb N \langle M \rangle$ of finite multisets $r_1, r_2 \in \mathbb N \langle M \rangle$ is defined by $(r_1 \oplus r_2)(m) = r_1(m) + r_2(m)$
for all $m \in M$. We define the {\it Cauchy product} $(r_1 \cdot r_2) \in \mathbb N \langle M \rangle$ of two finite multisets $r_1, r_2 \in \mathbb N \langle M \rangle$ by
$$(r_1 \cdot r_2)(m) = \sum \left( r_1(m_1) \cdot r_2(m_2) \; | \; m_1, m_2 \in M, m_1 \diamond m_2 = m \right).$$ Note that in the equation above there are finitely many non-zero summands, because the  multisets $r_1$ and $r_2$ are finite.
Let $n \ge 1$ and $r_1, ..., r_n \in \mathbb N \langle M \rangle$. We also define the {\it valuation} $\Val(r_1, ..., r_n) \in \mathbb N \langle M \rangle$ by
$$\Val(r_1, ..., r_n)(m) =  \sum \left( \prod\nolimits_{i = 1}^n r_i (m_i) \; | \;  m_1, ..., m_n \in M, \Val(m_1, ...,  m_n) = m \right).$$ Note that the right side of the equation above also contains only finitely many non-zero summands. 
 The {\it empty multiset} $\varepsilon$ is the finite multiset whose support is empty. A {\it simple multiset} over $M$ is a finite multiset $r \in \mathbb N \langle M \rangle$ such that $\supp(r) = \{m_r\}$ and $r(m_r) = 1$, so $r(m) = 0$ for all $m \neq m_r$. We denote such a simple multiset $r$ by $[m_r]$.
The collection of all simple multisets over $M$ is denoted by $\Mon(M)$.

As opposed to the case of pv-monoids \cite{DM12}, the pv-structure $\mathcal M$ does not contain a commutative and associative sum operation to define the semantics of the disjunction and the existential quantification. For this, we employ the sum of multisets. Let $\varphi$ be a multi-weighted formula over $\Sigma$ and $\mathcal M$, and ${\mV \supseteq \Free(\varphi)}$.  We define the auxiliary multiset semantics function ${\langle \varphi \rangle_{\mV}: \Sigma_{\mV}^+ \to \Mult}$ relying on the ideas of \cite{DG09} (cf. also \cite{DM12}) as follows: for all $(w, \sigma) \notin \mathcal N_{\mV}$, $\langle \varphi \rangle_{\mV} (w, \sigma) = \varepsilon$ and, for all ${(w, \sigma) \in \mathcal N_{\mV}}$, $\langle \varphi \rangle_{\mV}(w, \sigma)$ is defined inductively as shown in Table \ref{Table:Semantics}. 
\begin{table}[t]
{\scriptsize
{\begin{equation}
\begin{aligned}
\langle m \rangle _{\mathcal V}(w, \! \sigma) \! &= \! [m] \\
\langle P_a(x) \rangle _{\mathcal V}(w, \!\sigma) \! &= \! \begin{cases} \! [\one], \! \! \! \! & \mbox{if } w_{\sigma(x)} \! = \! a,\\  \, \! \varepsilon, & {\rm otherwise} \end{cases} \\
\langle x \! \le \! y \rangle_{\mathcal V}(w,\! \sigma) \! &= \! \begin{cases} \! [\one], \! \! \! \! & \mbox{if } \sigma(x) \le \sigma(y), \\ \varepsilon, & {\rm otherwise} \end{cases} \\
\langle  x \! \in \! X \rangle_{\mathcal V}(w,\! \sigma) \! &= \! \begin{cases} \! [\one], \! \! \! \! & \mbox{if } \sigma(x) \in \sigma(X), \\ \varepsilon, & {\rm otherwise} \end{cases} \\
\langle \lnot \beta \rangle_{\mathcal V}(w, \!\sigma) \! &= \! \begin{cases} \! [\one], \! \! \! \!& \mbox{if } \langle \beta \rangle_{\mathcal V}(w, \! \sigma) \! =\! \varepsilon, \\ \varepsilon, & {\rm otherwise} \end{cases} 
\end{aligned} \quad
\begin{aligned}
\langle \varphi_1 \! \vee \! \varphi_2 \rangle_{\mathcal V}(w, \! \sigma)\! &=\! \langle \varphi_1 \rangle_{\mathcal V}(w, \! \sigma)\! \oplus \!\langle \varphi_2 \rangle_{\mathcal V}(w, \! \sigma) \\
\langle \varphi_1 \! \wedge \! \varphi_2 \rangle_{\mathcal V}(w, \! \sigma) \! &= \! \langle \varphi_1 \rangle_{\mathcal V}(w, \! \sigma) \! \cdot \! \langle \varphi_2 \rangle_{\mathcal V}(w, \!\sigma) \\
\langle \exists x \varphi \rangle_{\mathcal V}(w, \! \sigma) \! &=  \!\! \! \! \! \! \! \bigoplus\limits_{i \in \dom(w)} \! \! \! \! \! \langle \varphi \rangle_{\mathcal V \cup \{x\}} (w, \sigma[x/ i]) \\
\langle \exists X  \!\varphi \rangle_{\mathcal V}(w, \!\sigma) \! &= \! \! \! \! \! \bigoplus\limits_{I  \subseteq  \dom( \!w \!)}  \!\!\!\!\langle \varphi \rangle_{\mathcal V \cup \{X\}} (w, \sigma[X / I])\\
\langle \forall x \varphi \rangle_{\mathcal V}(w, \! \sigma) \! &= \! \Val \! \left( \!\langle \varphi \rangle_{\mathcal V \cup  \{ \!x\!\}}(w, \! \sigma[x / i] ) \!\right)_{i \in \dom(w)} \\
\langle \forall X \beta \rangle_{\mathcal V} (w, \! \sigma) \! &= \! \Val \! \left( \!\langle \beta \rangle_{\mathcal V \cup \{ \! X\!\}} \!(w, \! \sigma[X / I]) \right)_{\! I  \subseteq  \dom( \!w \!)} 
\end{aligned} \nonumber
\end{equation}
} 
}
\caption{The auxiliary multiset semantics of multi-weighted MSO formulas over a pv-structure}
\label{Table:Semantics}
\end{table}
Here, $x,y,X \in \mathcal V, a \in \Sigma$, $m \in M$, $\beta$ is a boolean formula and $\varphi, \varphi_1, \varphi_2$ are multi-weighted formulas. In Table \ref{Table:Semantics}, for the semantics of $\forall X \varphi$ the subsets $I \subseteq \dom(w)$ are enumerated in some fixed order, e.g. lexicographically. For a formula $\varphi$, we put $\langle \varphi \rangle = \langle \varphi \rangle_{\Free(\varphi)}$. Then, we define the {\it semantics} $\li \varphi \re: \Sigma_{\varphi}^+ \to K$ as the composition $\li \varphi \re = \Phi \circ \langle \varphi \rangle$.

\begin{example}
\label{Example:GeometryLogics}
Let $A$ be an object on the plane whose displacement is managed by two types of commands: $\leftrightarrow$ and $\updownarrow$. After receiving the command $\leftrightarrow$ the object moves one step to the left or to the right; after receiving $\updownarrow$ one step up or down. Consider the $\mathbb R$-valuation structure $(\mathbb R^2, \Val, \Phi)$ from Example \ref{Example:F4}. We define $\diamond$ as the componentwise sum of vectors and put $\one = (0, 0)$. Then, $\mathcal M = (\mathbb R^2, \Val, \diamond, \one, \Phi)$ is an $\mathbb R$-pv-structure. Consider the following multi-weighted MSO sentence over the alphabet $\Sigma = \{\leftrightarrow, \updownarrow\}$ and the $\mathbb R$-pv-structure $\mathcal M$:
$$\varphi = \forall x ((P_{\leftrightarrow}(x) \to ({ (-1,0)} \vee (1,0))) \wedge (P_{\updownarrow}(x) \to ((0,-1)\vee (0,1))))$$ where, for a boolean formula $\varphi$ and a multi-weighted formula $\psi$, $\beta \to \psi$ is an abbreviation for $(\beta \wedge \psi) \vee \lnot \beta$.
For every sequence of commands $w \in \Sigma^+$, the multiset $\langle \varphi \rangle(w)$ contains all possible displacement vectors of the object. For example, let $w = \leftrightarrow \leftrightarrow$. The object has 4 possibilities to move: 1) two steps to the right; 2) one step to the right and then to the home position; 3) one step to the left and then to the home position; 4) two steps to the left. Then $\langle \varphi \rangle(w) = [(2, 0), (0, 0), (0, 0), (-2, 0)]$. The average displacement of the object is given by $\li \varphi \re$ for each sequence of commands $w$. For example, $\li \varphi \re(\leftrightarrow \leftrightarrow) = 1$, $\li \varphi\re (\leftrightarrow \updownarrow) = \sqrt{2}.$ 
\end{example}
\begin{comment}
We say that the pv-structure $\mathcal M$ is {\it left-multiplicative} if, for all $n \ge 1$ and $m, m_1, ..., m_n \in M$: $m \diamond \Val(m_1, ..., m_n) = \Val(m \diamond m_1, m_2, ..., m_n)$. We call $\mathcal M$ $\Val$-{\it commutative} if for any two sequences $(m_1, ..., m_n)$ and $(m_1', ..., m_n')$ in $M$, we have $\Val(m_1, ..., m_n) \diamond \Val(m_1', ..., m_n') = \Val(m_1 \diamond m_1', ..., m_n \diamond m_n')$. We call $\mathcal M$ {\it associative} if $\diamond$ is associative. Similar requirements also were important in \cite{DM12}.

Consider the valuation structures of Examples \ref{Example:F1}, \ref{Example:F2} and \ref{Example:F3}. If we define $\diamond$ as the componentwise sum of vectors, then the corresponding pv-structures are associative, $\Val$-commutative and left-multiplicative. Next, consider the valuation structure of Example \ref{Example:F3}. If we define $\diamond$ by $(x_1, t_1) \diamond (x_2, t_2) = (x_1 + \lambda^{t_1} x_2, t_1 + t_2)$, then the corresponding pv-monoid is associative, left-multiplicative, but not $\Val$-commutative.
\end{comment}

Note that the multi-weighted MSO logic over $K$-pv-structures contains the case of weighted MSO logic over semirings (cf. \cite{DG07, DG09}). Hence, in general, multi-weighted MSO logic is expressively more powerful than multi-weighted automata. 
%For counterexamples, we refer to \cite{DG09}.

Our main result for finite words is the following theorem.

\begin{theorem}
\label{Theorem:Fin}
Let $\Sigma$ be an alphabet, $K$ a set, $\mathcal M\! = \!(M, \Val, \diamond, \one, \Phi)$ a $K$-pv-structure and $s\!: \Sigma^+ \!\! \to \! K$. Then $s = ||\mathcal A||$ for some multi-weighted automaton \!$\mathcal A$ over $\Sigma$ and $\mathcal M$ iff $s = \li \varphi \re$ for a syntactically restricted multi-weighted MSO sentence $\varphi$ over $\Sigma$ and $\mathcal M$.
\end{theorem}
The proof is similar to the proof of the corresponding Theorem \ref{Theorem:Inf} for infinite words. For lack of space, we skip it.

We consider examples of decision problems for multi-weighted MSO logic.

\begin{example}
\label{Example:Emp1}
Let $\Sigma$ be an alphabet and ${\mathcal M = (\mathbb Q \times \mathbb Q_{\ge 0}, \Val, \diamond, (0, 0), \Phi)}$ the $\overline{\mathbb R}$-pv-structure where $\diamond$ is the componentwise sum, and $\Val$ and $\Phi$ are defined as in Example \ref{Example:F1}. Let $\varphi$ be a multi-weighted MSO sentence over $\Sigma$ and $\mathcal M$, and $\nu \in \mathbb Q$ a thres\-hold. The {\it $\ge\!\nu$-emptiness problem} is whether there exists a word $w \in \Sigma^+$ such that $\li \varphi \re(w) \! \ge \! \nu$. If $\varphi$ is syntactically restricted, then, using our Theorem \ref{Theorem:Fin}, we can effectively translate $\varphi$ into a multi-weighted automaton over $\Sigma$ and $\mathcal M$. Then $\ge  \!\! \nu$-emptiness for these multi-weighted automata can be decided in the following way. First, we use a shortest path algorithm to decide whether there exists a path with cost $0$, i.e. $||\mathcal A ||(w)\! =\! \infty \! \ge \! \nu$ for some $w$. If this is not the case (i.e. the costs of all accepting paths in $\mathcal A$ are strictly positive), we use the same technique as for the $\ge \! \! \nu$-emptiness problem for ratio automata with strictly positive costs (cf. \cite{FGR11}, Theorem 3). We replace the weight $(r, c)$ of every transition by the single value $r - \nu c$ and obtain a weighted automaton $\mathcal A'$ over the max-plus semiring $\mathbb Q \cup \{-\infty\}$. Then, $||\mathcal A||(w) \ge \nu$ iff the semiring-behavior of $\mathcal A'$ on $w$ is not less than zero. Then, the decidability of our problem follows from the decidability of the $\ge\!\! 0$-emptiness problem for max-plus automata.
\end{example}

\begin{example}
Let $\Sigma$ be an alphabet and $\mathcal M = (\mathbb Q^2, \Val, \diamond, (0, 0), \Phi)$ where $\diamond$ is the componentwise sum, and $\Val$ and $\Phi$ are as in Example \ref{Example:F2}. Again, using our Theorem \ref{Theorem:Fin}, we can reduce the $\le\!\! \nu$-emptiness problem (defined similarly as in Example \ref{Example:Emp1}) for syntactically restricted multi-weighted MSO logic over $\Sigma$ and $\mathcal M$ to the emptiness problem for multi-weighted automata. This problem is decidable, since the optimal conditional reachability for multi-priced timed automata is decidable \cite{LR05}.
\end{example}

\section{Multi-weighted Automata and MSO Logic on Infinite Words}

In this section, we develop a general model for both multi-weighted automata and MSO logic on infinite words. Recall that, for a set $M$, $\eN \li M \re$ is the collection of all multisets over $M$. Let $M^{\omega}$ denote the set of all $\omega$-infinite words over $M$.

\begin{definition}
Let $K$ be a set. A {\rm product $K$-$\omega$-valuation structure ($K$-$\omega$-pv structure)} is a tuple $(M, \valo, \diamond, \one, \Phi)$ where  
\begin{itemize}
\item $M$ is a set, $\one \in M$ and $\Phi: \eN \li M \re \to K$;
\item $\valo: M^{\omega} \to M$ with $\valo(m \one^{\omega}) = m$ \,for all $m \in M$;
\item $\diamond: M \times M \to M$ such that $m \diamond \one = \one \diamond m = m$  for all $m \in M$.
\end{itemize}
\end{definition}
A {\it Muller automaton} over an alphabet $\Sigma$ is a tuple $\mathcal A = (Q, I, T, \mathcal F)$ where $Q$ is a set of states, $I \subseteq Q$ is a set of initial states, $T \subseteq Q \times \Sigma \times Q$ is a transition relation and $\mathcal F \subseteq 2^Q$ is a Muller acceptance condition. {\it Infinite paths} ${\pi = (t_i)_{i \in \omega}}$ of $\mathcal A$ are defined as infinite sequences of matching transitions, say ${t_i = (q_i, a_i, q_{i+1})}$. Then we call the word $w = (a_i)_{i \in \omega}$ the {\it label} of the path $\pi$ and $\pi$ a path on $w$. We say that a path $\pi = (q_i, a_i, q_{i+1})_{i \in \omega}$ is {\it accepting} if $q_0 \in I$ and $\{{q \in Q} \; | \; {q = q_i} \text{ for infinitely many } {i \in \omega}\} \in \mathcal F$. Let $\Acc_{\mathcal A} (w) $ denote the set of all accepting paths of $\mathcal A$ on $w$.

For the rest of this section, we fix an alphabet $\Sigma$ and a $K$-$\omega$-pv structure $ {\mathcal M = (M, \valo, \diamond, \one, \Phi)}$.

\begin{definition}
A {\rm multi-weighted Muller automaton} over $\Sigma$ and $\mathcal M$ is a tuple ${\mathcal A = (Q, I, T, \mathcal F, \gamma)}$ where $(Q, I, T, \mathcal F)$ is a Muller automaton and $\gamma: T \to M$.
\end{definition}
Let $\mathcal A$ be a multi-weighted Muller automaton over $\Sigma$ and $\mathcal M$, $w \in \so$ and $\pi = (t_i)_{i \in \omega}$ an accepting path on $w$. The {\it weight} of $\pi$ is defined by $\Weight_{\mathcal A}(\pi) = \valo(\gamma(t_i))_{i \in \omega}$. Let $|\mathcal A | (w) \in \eN \li M \re$ be the multiset containing the weights of paths in $\Acc_{\mathcal A}(w)$. Formally, ${|\mathcal A| (w)(m) = |\{\pi \in \Acc_{\mathcal A}(w) \; | \; \Weight_{\mathcal A}(w) = m \}}|$ where, for an infinite set $X$, we put $|X| = \infty$. The {\it behavior} of $\mathcal A$ is the $\omega$-series $||\mathcal A||: \so \to K$ defined by $||\mathcal A||(w) = \Phi(| \mathcal A | (w))$.

\begin{remark}
The multiplication $\diamond$, the unital element $\one$ and the condition ${\valo(m \one^{\omega}) = m}$ are irrelevant for the definition of the behaviors of multi-weighted automata. However, they will be used to describe the semantics of multi-weighted MSO formulas.
\end{remark}
We consider several examples of multi-weighted automata $\mathcal A$ over $\Sigma$ and $\mathcal M$, and their behaviors.

\begin{example}
\label{Example:I1}
Consider the reward-cost ratio setting of Example \ref{Example:F1} for infinite words. 
For a sequence $(r_i, c_i)_{i \in \omega} \in (\mathbb R \times \mathbb R_{\ge 0})^{\omega}$ of reward-cost pairs, the {\it supremum ratio} (cf. \cite{BBL04}) is defined by $\limsup\limits_{n \to \infty} \frac{\sum_{i = 0}^n r_i}{\sum_{i = 0}^n c_i} \in \overline{\mathbb R}$ where $\frac{r}{0} = \infty$. Unfortunately, since $\sum_{i = 0}^{\infty} r_i$ and $\sum_{i = 0}^{\infty} c_i$ may not exist or may be infinite, we cannot proceed as for finite words by considering pairs of accumulated rewards and costs and their ratios. Instead, we can define $\mathcal M$ as follows.  Let $M = \overline{\mathbb R} \times \mathbb R_{\ge 0}$, $K = \overline{\mathbb R}$ and $\one = (0, 0)$. Let $\mu = (r_i, c_i)_{i \in \omega} \in (\mathbb R \times \mathbb R_{\ge 0})^{\omega}$. If $\sum_{i = 0}^{\infty} r_i$ and $\sum_{i = 0}^{\infty} c_i$ are finite, then we put $\valo(\mu) = (\sum_{i = 0}^{\infty} r_i, \sum_{i = 0}^{\infty} c_i)$. Otherwise, we put $\valo(\mu) = \left( \limsup\limits_{n \to \infty} \frac{\sum_{i = 0}^n r_i}{\sum_{i = 0}^n c_i}, 1 \right)$. For sequences $\mu \in M^{\omega} \setminus (\mathbb R \times \mathbb R_{\ge 0})^{\omega}$, we define $\valo(\mu)$ arbitrarily keeping $\valo(m \one^{\omega}) = m$. Let also $\diamond$ be the componentwise sum where $\infty \! + \! (\text{-}\infty)$ is defined arbitrarily. The evaluator function $\Phi$ is defined by $\Phi(r) = \sup\limits_{(x, y) \in \supp(r)} \frac{x}{y}$. Then, $||\mathcal A||(w)$ is the maximal supremum ratio of accepting paths of $w$. The corresponding model for timed automata was considered in \cite{BBL04, BBL08}.
\end{example}

\begin{example}
\label{Example:I2}
Let $E_{\max} = (E_{\max}^1, ..., E_{\max}^n) \in \mathbb Z^n$ where $E_{\max}^i > 0$ for all $i$, and $M = [-E_{\max}, E_{\max}] \subseteq \mathbb Z^n$, i.e. $M$ consists of all vectors $(v^1, ..., v^n) \in \mathbb Z^n$ such that $-E_{\max}^i \le v^i \le E_{\max}^i$  for each $i \in \{1, ..., n\}$. Let $K = \mathbb B = \{\text{false}, \text{true}\}$, the boolean semiring and ${\one = (0, ..., 0)}$. For $u_1 = (u_1^1, ..., u_1^n)$ and $u_2 = (u_2^1, ..., u_2^n) \in M$, we put $u_1 \diamond u_2 = (v^1, ..., v^n)$ where $v^i = \max \{\min\{u_1^i + u_2^i, E_{\max}^i \}, -E_{\max}^i\}$. For $(m_i)_{i \in \omega} \in M^{\omega}$ we define the sequence $(v_i)_{i \in \omega}$ in $M$ as follows. We put $v_0 = (0, ..., 0)$ and $v_{i+1} = v_i \diamond m_i$ for all $i \in \omega$. Then, let $\valo((m_i)_{i \in \omega}) = (x^1, ..., x^n) \in M$ where $x^j = \inf \{v_i^j \; | \; i \in \omega\}$ for all $1 \le j \le n$. 
Let $\Phi$ be defined by $\Phi(r) =  \text{true}$ iff there exists ${(m^1, ..., m^n) \in \supp(r)}$ with $m^j \ge 0$ for all $1 \le j \le n$. This model corresponds to the one-player energy games considered in \cite{FJLS11}.
\end{example}
The syntax of the {\it multi-weighted MSO logic} over $\Sigma$ and $\mathcal M$ is defined exactly as for finite words (cf. Section 3). To define the semantics of this logic, we proceed similarly as for finite words, i.e. by means of the auxiliary multiset semantics. For this, we consider $\eN$ as the totally complete semiring $(\eN, +, \cdot, 0, 1)$ (cf. \cite{DG09}) where $0 \cdot \infty \! = \! \infty \cdot 0 \! = \! 0$. The sum $\oplus$ and the Cauchy product $\cdot$ for infinite multisets from $\eN \li M \re$ are defined as for finite words. The $\omega$-valuation $\valo(r_i)_{i \in \omega}$ for $r_i \in \eN \li M \re$ is defined for all $m \in M$ by
$$\valo((r_i)_{i \in \omega})(m) = \sum \left( \prod\nolimits_{i \in \omega} r_i (m_i) \; | \; (m_i)_{i \in \omega} \in M^{\omega} \text{ and } \valo(m_i)_{i \in \omega} = m \right).$$
The {\it empty multiset} $\varepsilon \in \eN \li M \re$ and {\it simple multisets} $[m] \in \eN \li M \re$ (for $m \in M$) are defined in the same way as for finite words. Let $\Mon(M) = \{[m] \; | \; m \in M\}$.

Let $\varphi$ be a multi-weighted MSO formula over $\Sigma$ and $\mathcal M$, and $\mV \supseteq \Free(\varphi)$. We define the auxiliary multiset semantics $\langle \varphi \rangle_{\mV}: \sov \to \eN \li M \re$ inductively on the structure of $\varphi$ as in Table \ref{Table:Semantics} where we have to replace $\Val$ by $\Val^{\omega}$. For $w \in \so$, we let $\dom(w) = \omega$. To define the semantics $\langle \forall X \varphi \rangle$, we have to extend $\valo$ for multisets to index sets of size continuum such that $\valo((r_i)_{i \in I}) = \varepsilon$ whenever $r_i = \varepsilon$ for some $i \in I$, and $\valo(([\one])_{i \in I}) = [\one]$. The {\it semantics} of $\varphi$ is defined by $\li \varphi \re = \Phi \circ \langle \varphi \rangle$.

\begin{example}
Assume that a bus can operate two routes A and B which start and end at the same place. The route R lasts $t_{\rm R}$ time units and profits $p_{\rm R}$ money units on the average per trip, for ${\rm R} \in \{{\rm A}, {\rm B}\}$. We may be interested in making an infinite schedule for this bus which is represented as an infinite sequence from $\{{\rm A,B}\}^{\omega}$. This schedule must be fair in the sense that both routes A and B must occur infinitely often in this timetable (even if the route A or B is unprofitable). The optimality of the schedule is also preferred (we wish to profit per time unit as much as possible). We consider the $K$-$\omega$-pv structure $\mathcal M$ from Example \ref{Example:I1} and a one-element alphabet $\Sigma = \{\tau\}$ which is irrelevant here.
Now we construct a weighted MSO sentence $\varphi$ over $\Sigma$ and $\mathcal M$ to define the optimal income of the bus per time unit (supremum ratio between rewards and time):
$$
\varphi = \exists X \! \! \left(\stackrel{\infty}{\exists}\!\!x(x \! \in \! X \!) \; \wedge \stackrel{\infty}{\exists}\!\!x(x \! \notin \! X \!) \wedge  
\forall x ((x \! \in \! X \! \to \! (p_{\rm A}, t_{\rm A})) \wedge (x \! \notin \! X \! \to \! (p_{\rm B}, t_{\rm B}))\right)
$$
where $\stackrel{\infty}{\exists}\!\!x \psi$ is an abbreviation for a boolean formula $\forall y( \lnot \forall x( \lnot (y \le x \wedge \psi)))$. Here, the second order variable $X$ corresponds to the set of positions in an infinite schedule which can be assigned to the route A. Then,
$$|\varphi|(\tau^{\omega}) = \sup \left \{ \limsup_{n \to \infty} \frac{p_{\rm A} \cdot |I \cap \overline{n}| + p_{\rm B} \cdot |I^c \cap \overline{n}| } {t_{\rm A} \cdot |I \cap \overline{n}| + t_{\rm B} \cdot |I^c \cap \overline{n}|} \; | \; I \subseteq \mathbb N \text{ with } I, I^c \text{ infinite}\right \}$$
where $\overline{n} = \{0, ..., n\}$ and $I^c = \mathbb N \setminus I$.
\end{example}
Now we state our main result for infinite words.
\begin{theorem}
\label{Theorem:Inf}
Let $\Sigma$ be an alphabet, $K$ a set and $\mathcal M = (M, \valo, \diamond, \one, \Phi)$ a $K$-$\omega$-pv structure. Let $s: \so \to K$ be an $\omega$-series. Then
$s = ||\mathcal A||$ for some multi-weighted Muller automaton $\mathcal A$ over $\Sigma$ and $\mathcal M$ iff $s = \li \varphi \re$ for some syntactically restricted multi-weighted MSO sentence $\varphi$ over $\Sigma$ and $\mathcal M$.
\end{theorem}
In the rest of this section, we give the proof idea of this theorem. Let $\mwMA(\Sigma, \mathcal M)$ denote the collection of all multi-weighted Muller automata over $\Sigma$ and $\mathcal M$. Let ${\mathcal A \in \mwMA(\Sigma, \mathcal M)}$. We can consider $|\mathcal A|$ as an $\omega$-series $| \mathcal A |: \so \to \eN \li M \re$. We call $|\mathcal A|$ the {\it multiset-behavior} of $\mathcal A$. Then $||\mathcal A|| = \Phi \circ | \mathcal A |$. 
%Let $\li \mwMA(\Sigma, \mathcal M) \re$ denote the collection $\{ \li \mathcal A \re \; | \; \mathcal A \in \mwMA(\mathcal M, \Sigma)\}$. 
Let $\mwMSO^{\text{res}}(\Sigma, \mathcal M)$ denote the set of all syntactically rest\-ricted multi-weighted MSO sentences over $\Sigma$ and $\mathcal M$. %Let $\li \mwMSO^{\text{res}} (\mathcal M, \Sigma) \re = \{ \li \varphi \re \; | \; \varphi \in \mwMSO^{\text{res}} (\mathcal M, \Sigma)\}$. 
Since, for any multi-weighted formula $\varphi$, $\li \varphi \re = \Phi \circ \langle \varphi \rangle$, it suffices to show that $\mwMA(\Sigma, \mathcal M)$ with the multiset-behavior and $\mwMSO^{\text{res}} (\Sigma, \mathcal M)$ with the multiset-semantics are expressively equivalent. 

For this, we can show that $(\eN \li M \re, \oplus, \valo, \cdot, \varepsilon, [\one])$ is an $\omega$-pv monoid as defined in \cite{DM12}. Let $D \subseteq \eN \li M \re$. We denote by $\wMA(\Sigma, D)$ the collection of weighted automata over $\Sigma$ and the $\omega$-pv monoid $\eN \li M \re$ where the weights of transitions are taken from $D$. 
%Let $\mathcal A $. Let $||\mathcal A||: \so \to \eN \li M \re$ denote the behavior of $\mathcal A$ as defined in \cite{DM12} and $||\wMA(D, \Sigma)|| = \{ ||\mathcal A|| \; | \; \mathcal A \in \wMA(D, \Sigma)\}$. 
Let $\wMSO^{\text{res}}(\Sigma, D)$ denote the set of syntactically restricted sentences over $\Sigma$ and the $\omega$-pv monoid $\eN \li M \re$ with constants from $D$. Let $[\![\varphi]\!]$ denote the semantics of $\varphi \in \wMSO^{\text{res}} (\Sigma, \mathcal M)$ as defined in \cite{DM12}. The proof scheme of our result is depicted in Fig. \ref{Fig:Proof_Inf}. Here, $\leftrightarrow$ means the expressive equivalence and $\to$ the expressive inclusion. 

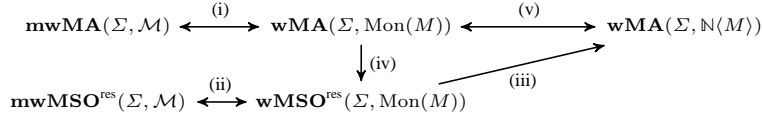
\begin{figure}[t]
\begin{center}
\begin{tikzpicture}[->,>=stealth',shorten >=1pt,auto,node distance=4cm,
                    semithick]
  \tikzstyle{every state}=[minimum size = 0.6 cm, fill=none,draw=black,text=black]

  \node (A) at (0, 0)                   {{\scriptsize $\mwMA(\Sigma, \mathcal M)$}};
  \node (B) at (3.5, 0)     {{\scriptsize $\wMA(\Sigma, \Mon(M))$}};
  \node         (C) at (7.8, 0)  {{\scriptsize $\wMA(\Sigma, \Mult)$}};
  \node (D) at (3.5, -1) {{\scriptsize $\wMSO^{\text{res}} (\Sigma, \Mon(M))$}};
  \node (E) at (0, -1) {{\scriptsize $\mwMSO^{\text{res}} (\Sigma, \mathcal M)$}};
%   \path 
%(A) edge[line width = 2] node[left]{}  (B)
%(B) edge[line width = 2] node[right]{} (D)
%(A) edge node[left]{}  (C)
%(C) edge node[right]{}  (D);
%\node (E) at (0,1.7) {};
\path
(A) edge[<->]  node[above]{{\scriptsize (i) }} (B)
(D) edge[<->]  node[above]{{\scriptsize (ii) }} (E)
(B) edge[<->]  node[above]{{\scriptsize (v) }} (C)
(D) edge[->]  node[below]{{\scriptsize (iii) }} (C)
(B) edge[->]  node[right]{{\scriptsize (iv) }} (D);
%\draw[<->] (D) to (E);
%\draw[<->] (B) to (C);
%\draw[->] (D) to (C); 
%\draw[->] [bend left] (B) to (D); 
\end{tikzpicture}
\caption{The proof scheme of Theorem \ref{Theorem:Inf}}
\label{Fig:Proof_Inf}
\end{center}
\end{figure}

\begin{itemize}
\item [(i)] If we replace the weight $m \in M$ of every transition of a multi-weighted automaton $\mathcal A$ by the simple multiset $[m] \in \Mon(M)$, we obtain a weighted automaton $\mathcal A'$ over the pv monoid $\Mult$ such that the pv-monoid behavior of $\mathcal A'$ is equal to $| \mathcal A |$. Conversely, we can replace the weights $[m]$ in $\mathcal A'$ by $m$ to obtain a multi-weighted automaton with the same behavior.
\item [(ii)] Similarly to (i), we replace the constants $m$ occurring in MSO formulas by simple multisets $[m]$ and vice versa.
\item [(iii)] The proof is based on the proof of Theorem 6.2 (a) of Droste and Meinecke \cite{DM12}. We proceed inductively on the structure of $\varphi \in \wMSO^{\text{res}}( \Sigma, \Mon(M))$. Using the property $\valo(m \one^{\omega}) = m$ for $m \in M$, we show that every almost boolean formula is equivalent to a weighted Muller automaton with weights from $\Mon(M) \subseteq \Mult$. 
Let $\varphi, \varphi_1$ and $\varphi_2$ be weighted MSO formulas with constants from $\Mon(M)$ such that $[\![\varphi]\!]$, $[\![\varphi_1]\!]$ and $[\![\varphi_2]\!]$ are recognizable by weighted Muller automata with weights from $D \subseteq \eN \li M \re$. Let  $\beta$ be any boolean formula. It can be shown that $[\![\varphi_1 \vee \varphi_2]\!]$, $[\![\exists x \varphi ]\!]$, $[\![\exists X \varphi]\!]$ and $[\![\varphi \wedge \beta]\!] = [\![\beta \wedge \varphi]\!]$ are also recognizable by weighted Muller automata with weights from $D$. If $\varphi$ is almost boolean, then $[\![\varphi]\!]$ is an $\omega$-recognizable step function with coefficients from $\Mult$. Using the construction of Lemma 8.11 of \cite{DG09}, cf. Theorem 6.2 of \cite{DM12}, we establish that $[\![\forall x \varphi]\!]$ is recognizable by a weighted automaton with weights from $\Mult$.
\item [(iv)] The proof follows from Theorem 6.2 of \cite{DM12} where a weighted automaton with weights in $D \! \subseteq \! \eN \li  M \re$\! was translated into an MSO sentence 
%in $\wMSO^{\text{res}}(D, \Sigma)$.
 with weights in $D$.
\item [(v)] 
Let $\mathcal A = (Q, I, T, \mathcal F, \gamma) \in \wMA(\Sigma, \Mult)$. We construct an auto\-maton ${\mathcal A' = (Q', I', T', \mathcal F', \gamma') \in \wMA(\Sigma, \Mon(M))}$ with the same behavior by unfolding each single transition of $\mathcal A$ labeled by a finite multiset into several transitions labeled by elements of this multiset as follows. 
\begin{itemize}
\item 
$Q' = I \cup \{(q, m, i): t = (p, a, q) \in T, m \in \supp(\gamma(t)),  1 \le i \le \gamma(t)(m)\}$
\item  $I' = I$, ${\mathcal F' = \{\{(q_1, m_1, k_1), ..., (q_n, m_n, k_n)\} \subseteq Q' \setminus I \; | \; \{q_1, ..., q_n\} \in \mathcal F\}}$.
\item $T' = T_1 \cup T_2$, where $T_1$ consists of all transitions $(p, a, (q, m, i))$ from $I \times \Sigma \times (Q' \setminus I)$ with $(p,a,q) \in T$; $T_2$ consists of all transitions $((q_1, m_1, i_1), a, (q_2, m_2, i_2))$ from $(Q' \setminus I) \times \Sigma \times (Q' \setminus I)$ with $(q_1, a, q_2) \in T$.
\item  For all $t = (q', a, (q,m,i)) \in T'$, let $\gamma'(t) = [m]$.
\end{itemize}
\end{itemize}

\section{Conclusion}

We have extended the use of weighted MSO logic to a new class of multi-weighted settings. 
We just note that, as in \cite{DM12}, for $K$-pv-structures and $K$-$\omega$-pv structures with additional properties there are larger fragments of multi-weighted MSO logic which are still expressively equivalent to multi-weighted automata. 
Since our translations from formulas to automata are effective, we can reduce the decidability problems for multi-weighted logics to the corresponding problems for multi-weighted automata. Decidability results of, e.g., \cite{BBL04, FJLS11, FGR11, LR05} lead to decidability results for multi-weighted nondeterministic automata. However, for infinite words, the authors did not consider Muller acceptance condition for automata. Therefore, our future work will investigate decision problems for multi-weighted Muller automata. Also, weighted MSO logic for weighted timed automata was investigated in \cite{Qua11}. In our further work, we wish to combine the ideas of \cite{Qua11} and the current work to obtain a B\"uchi theorem for multi-weighted timed automata.

\bibliographystyle{abbrv}
\bibliography{literature}

\end{document}